\documentclass[fleqn,10pt,twocolumn,letterpaper]{SICE_ISCS}
\usepackage{palatino,epsfig,fleqn,amssymb,cite,color}

\newtheorem{prop}{\bf Proposition}
\newtheorem{lemma}{\bf Lemma}

\def\QED{~\rule[-1pt]{5pt}{5pt}\par}

\DeclareMathAlphabet{\matheur}{U}{eur}{m}{n}
\DeclareMathAlphabet{\matheurb}{U}{eur}{b}{n}
\DeclareMathAlphabet{\matheus}{U}{eus}{m}{n}
\DeclareMathAlphabet{\matheuf}{U}{euf}{m}{n}

\newcommand{\bfThe}{\mathbf{\Theta}}
\newcommand{\bfdel}{\mbox{\boldmath$\delta$}}
\newcommand{\bfDel}{{\bf\Delta}}

\newcommand{\hs}{\hspace{4mm}}

\renewcommand{\t}{^{\mbox{\tiny\sf T}}}
\newcommand{\IC}{{\mathbb{C}}}
\newcommand{\IR}{{\mathbb{R}}}

\newcommand{\IU}{{\mathbb{U}}}
\newcommand{\IV}{{\mathbb{V}}}
\newcommand{\IW}{{\mathbb{W}}}
\newcommand{\II}{{\mathbb{I}}}

\newcommand{\diag}{{\rm diag}}

\newenvironment{mat}{\left[\begin{array}}{\end{array}\right]}

\newcommand{\rhoU}{{\bar\rho}}
\newcommand{\rhoL}{{\varrho}}

\DeclareMathAlphabet{\matheur}{U}{eur}{m}{n}
\DeclareMathAlphabet{\matheurb}{U}{eur}{b}{n}
\DeclareMathAlphabet{\matheus}{U}{eus}{m}{n}
\DeclareMathAlphabet{\matheuf}{U}{euf}{m}{n}

\newcommand{\LFT}{{\mathcal{F}}_\ell}
\newcommand{\calG}{{\mathcal{G}}}

\newcommand{\RHinf}{{\bf RH}_\infty}

\title
{Robust Instability Radius for Multi-agent Dynamical Systems \\ with Cyclic Structure}

\author{Shinji Hara${}^{1\dagger}$, Tetsuya Iwasaki${}^{2}$ and Yutaka Hori${}^{3}$} 
\speaker{Shinji Hara}

\affils{${}^{1}$
Systems and Control Engineering,  Tokyo Institute of Technology, Japan\\
{\tt shinji\_hara@ipc.i.u-tokyo.ac.jp}\\
${}^{2}$
Mechanical and Aerospace Engineering,  University of California Los Angeles,  USA\\
{\tt tiwasaki@ucla.edu}\\
${}^{3}$
Applied Physics and Physico-Informatics, Keio University, Japan\\ 
{\tt yhori@appi.keio.ac.jp}
}

\abstract{%
This paper is concerned with robust instability analysis  for linear multi-agent dynamical systems with cyclic structure. 
This relates to interesting and important periodic oscillation phenomena in biology and neuronal science, 
since the nonlinear phenomena often occur when the linearized model around an equilibrium point is unstable. 
We first make a problem setting on the analysis and define the notion of robust instability radius (RIR) as 
a quantitative measure for maximum allowable stable dynamic perturbation in terms of the $H$-infinity norm. 
After showing lower bounds of the RIR, we derive the exact RIR,  which is  analytic and scalable, 
for first order time-lag agents.    
Finally, we make a remark on the potential applicability to some classes of higher order systems.  
}  


\begin{document}

\maketitle

\section{Introduction} 
\label{sec:Intro}

There are a number of interesting and important periodic oscillation phenomena 
in biology such as Repressilator \cite{Elowitz2000} in synthetic biology, 
spike-type periodic signals in neuronal dynamics \cite{FHNmodel}, 
periodic pattern generation by Turing instability \cite{YMKH:MBMC2015}, and so on. 
Many of these cases are related to instability of the linearized model around an equilibrium point, and 
it is generally difficult to derive the exact mathematical models and reduced order approximate models 
are often utilized for the analysis.   
Hence, robust instability analysis against dynamic uncertainties is very important to analyze the persistence 
of oscillation phenomena theoretically.  

Motivated by this, the authors have proposed a robust instability problem as a new control problem 
\cite{HIH:LCSS2020, HIH:Automatica2020}.  
It should be emphasized that the robust instability analysis is similar to but quite different from
the robust stability analysis.   Actually, the former is a strong stabilization problem \cite{Youla:Automatica1974} 
to find a minimum norm stable perturbation that stabilizes a given unstable system  
when the uncertainty is modeled by a ball measured by the $H_\infty$ norm.  
This clearly indicates the difficulty of the problem. 
For example, the small gain condition in terms of the $L_\infty$ norm provides only 
a sufficient condition for the robust instability as seen in e.g., \cite{Inoue:ECC2013,Tsypkin:Automatica1994}. 
The reason is as follows.  
For robust stability analysis, the system instability can be detected when one of stable poles 
reaches the imaginary axis.  
However, for the case of robust instability, the system may remain unstable even if some of unstable poles 
reach the imaginary axis.  
Hence, the infinity norm does not work in general for robust instability analysis. 

Consequently, our main question here is "{\em Under what condition does the infinity-norm 
work as a measure even for the robust instability analysis?}" 
There is a partial answer to the question in \cite{HIH:LCSS2020, HIH:Automatica2020}, where   
some exact results on the robust instability radius for linear systems are shown.  
Also, the robust instability margin was investigated by taking into account the change of the equilibrium point 
caused by the static gain change for a class of nonlinear systems, and the effectiveness of the exact results 
is confirmed by applications to FitzHugh-Nagumo model for neuronal dynamics and Repressilator in synthetic biology.    

This paper is concerned with the robust instability problem for multi-agent dynamical systems with cyclic structure. 
A long-term goal of this type of research is to establish a  new theory for a class of nonlinear networked systems 
which covers several biomolecular systems in synthetic biology \cite{Alon2006}. 
However, this paper treats only the linear case as an important fundamental research topic. The main contribution 
of this paper is a characterization of the exact value of maximum allowable dynamic perturbation, 
or the robust instability radius (RIR), for the persistence of instability for multi-agent dynamical systems 
with any number of first order agents.  
The result is scalable as well as analytic, and hence we can handle multi-agent systems with any number of agents.   
This is in contrast with the authors' previous investigation  \cite{HIH:LCSS2020, HIH:Automatica2020}. 

The first key for the analysis is to show the equivalence of the classes of heterogeneous and homogeneous 
perturbations in terms of robust instability.  
In other words, we can show that the class of worst case perturbations  includes a subset of homogeneous perturbations.  
This nice property can make the problem simpler and yield useful lower bounds on the RIR. 
Then, we follow the two-step procedure of constructing an all-pass perturbation proposed 
in  \cite{HIH:LCSS2020, HIH:Automatica2020} to get the exact result for first order time-lag systems. 
It should be emphasized that the final RIR result is the same as the robust stability result in
\cite{HIH:Automatica2019} based on the generalized frequency variable framework \cite{HTI:TAC2014}.  

The remainder of this paper is organized as follows. 
The problem formulation and the preliminary analysis are presented in Section~\ref{sec:MultiAgentSytem}. 
Section~\ref{sec:GeneralCycNetwork} is devoted to the analysis of the RIR for cyclic network structure 
with multiplicative uncertainties, where some lower bounds of the RIR are shown. 
In Section~\ref{sec:ExactRIR}, we provide the exact RIR analysis for first order time-lag systems 
and remark the possible extensions. 
Section~\ref{sec:Concl} summarizes the contributions of this paper 
and addresses some future research directions.

We use the following notation. 
The set of proper stable real rational functions is denoted by $\RHinf$, 
and $\| \cdot \|_{H_\infty}$ represents the $H_\infty$ norm. 
$\| \cdot \|_{L_\infty}$ denotes the $L_\infty$ norm for rational functions.   
$\LFT(H,A)$ denotes the lower linear fractional transformation of $A$ by $H$:
\[
z=\LFT(H,A) w \; \Leftrightarrow \; 
\begin{mat}{c} z \\ y \end{mat} = H
\begin{mat}{c} w \\ u \end{mat}, \; 
u = A y.
\]
The set of integers $1,\ldots,n$ is denoted by $\II_n$.
\section{Robust Instability for Multi-agent Dynamical Systems}
\label{sec:MultiAgentSytem}
\subsection{Problem formulation}
\label{subsec:ProblemFormulation}

Consider the multi-agent system with $n$ uncertain dynamic SISO agents
described by
\begin{equation} \label{sys}
\begin{mat}{c} z \\ y \end{mat}=H(s)
\begin{mat}{c} w \\ u \end{mat}, \hs
\begin{array}{l} w=\Delta z, \\ u=Ay, \end{array} , 
\end{equation}
where $z(t),w(t),y(t),u(t)\in\IR^n$ are signals, $A\in\IR^{n\times n}$ is a constant matrix representing 
the interaction among the agents. 
We assume that $H(s)$ has a special form  
\[
H(s):=H_o(s)\otimes I_n 
\]
with a $2\times2$ transfer function $H_o(s)$.  
This  means that the multi-agent system is nominally homogeneous, and it is represented by 
 \[
 G(s):=\LFT(H(s),A) . 
 \]
$\Delta$ represents the stable uncertainties in the agents which are independent among the agents. 

We are interested in the condition under which the system is robustly unstable against 
$\Delta\in\bfDel\subset\RHinf^{n\times n}$ with a bound on the norm $\|\Delta\|_{H_\infty}$, where  
\begin{equation} \label{bfDel}
\bfDel:=\{\diag(\delta_1,\ldots,\delta_n):\,\delta_i\in\RHinf\,\}.
\end{equation}
Since the perturbed system is a feedback system consisting of $G(s)$ and $\Delta$, 
the system is said to be robustly unstable if the characteristic equation
\begin{equation} \label{CE}
\det(I-\Delta G(s))=0  
\end{equation}
has at least one root in the closed right half plane for all $\Delta\in\bfDel$
with the norm bound.
Clearly, the robust instability condition is violated when there exists 
$\Delta\in\bfDel$ that stabilizes (\ref{CE}) while satisfying the norm bound.  
Following the definition of the Robust Instability Radius (RIR)  in \cite{HIH:LCSS2020, HIH:Automatica2020}, 
the dynamic RIR (or simply RIR) $\rho_*$ is defined as the smallest magnitude of 
the perturbation that stabilizes the system: 
\begin{equation} \label{rir}
\hspace*{-4mm}
\begin{array}{l}
\rho_*:={\displaystyle\inf_{\Delta\in\bfDel}} \{~\|\Delta\|_{H_\infty}:~ \\
\hspace{10mm}\det(I-\Delta(s)G(s))=0 ~ \Rightarrow ~ \Re(s)<0~\}.
\end{array}
\end{equation}
The problem is to find a method for calculating $\rho_*$.

Let us introduce three related notions of RIRs to get the lower/upper bounds of the dynamic RIR. 
The dynamic RIR against homogeneous uncertainty $\rho_h$. :
\begin{equation} \label{hrir}
\hspace*{-2mm}
\begin{array}{l}
\rho_h:=\displaystyle\inf_{\delta\in\RHinf} ~\{~\|\delta\|_{H_\infty}:~ \\
\hspace{10mm} \det(I-\delta(s)G(s))=0 ~ \Rightarrow ~ \Re(s)<0~\}.
\end{array}
\end{equation}
The real/complex parametric RIRs, $\rho_r$ and $\rho_c$:
\begin{equation} \label{rrir}
\hspace*{-2mm}
\begin{array}{l} 
\rho_r:=\displaystyle \inf_{\Delta\in\bfDel_r} ~\{~\|\Delta\|:~ \\
\hspace{10mm}\det(I-\Delta G(s))=0 ~~ \Rightarrow ~~ \Re(s)<0~\} , 
\end{array}
\end{equation}
\begin{equation} \label{crir}
\hspace*{-2mm}
\begin{array}{l}
\rho_c:=\displaystyle \inf_{\Delta\in\bfDel_c} ~\{~\|\Delta\|:~ \\
\hspace{10mm}\det(I-\Delta G(s))=0 ~~ \Rightarrow ~~ \Re(s)<0~\} , 
\end{array}
\end{equation}
where $\bfDel_r$ and $\bfDel_c$ are  the set of $n\times n$ real diagonal matrices 
and the set of $n\times n$ complex diagonal matrices, respectively.

\subsection{Preliminary analysis}
\label{subsec:Preliminaries}

Without loss of generality, let us consider the case where $A$ is diagonalizable.\footnote{
When $A$ is not diagonalizable, an infinitesimal perturbation of $A$ will make it
diagonalizable without altering the robust instability radius due to 
continuous dependence of the characteristic roots on $A$.
}
Let the spectral decomposition of $A$ be given by
\[
A=T\Lambda T^{-1}, \hs \Lambda=\diag(\lambda_1,\ldots,\lambda_n).
\]
Then condition (\ref{CE}) is equivalent to
\begin{equation} \label{CE2}
\det(I-\hat\Delta(s)\hat G(s))=0, 
\end{equation}
where 
\[
\hat G(s):=\LFT(H(s),\Lambda), \;
\hat\Delta(s):=T^{-1}\Delta(s) T,
\]
because $G(s)=T\hat G(s)T^{-1}$ holds.  
Note that
\[
\hat G(s) := \diag(g_1(s),\ldots,g_n(s)) 
\]
with 
$g_i(s):=\LFT(H_o(s),\lambda_i), \; i\in\II_n$. 
Moreover, we have $\hat\Delta=\Delta$ if the perturbation is homogeneous,
i.e. $\Delta=\delta I$, %
if $A$ is normal and $T$ is unitary. These special cases in combination with the
small gain argument give the following result.

\medskip

\begin{prop} \label{prop:ul}
{\em
Consider the robust instability radius in (\ref{rir}) for system (\ref{sys}).
Suppose $G(s)$ is unstable and define
\begin{equation} \label{rhou}
\hspace*{-4mm}
\begin{array}{l}
\rhoU:= \displaystyle\inf_{\delta\in\RHinf} ~\{~\|\delta\|_{H_\infty}:~ 
1-\delta(s)g_i(s)=0 ~~ \\
\hspace{30mm} \Rightarrow ~~ \Re(s)<0,~~ \forall~i\in\II_n~\}
\end{array}
\end{equation}
\begin{equation} \label{rhol}
\rhoL_p:=\min_{i\in\II_n}~ 1/\|g_i\|_{L_\infty}.
\end{equation}
\begin{itemize}
\item[(a)] $\rho_*\leq\rhoU$ and $\rho_h=\rhoU$.
\item[(b)] $\rhoL_p\leq\rho_*$ if $A$ is normal, and  
$\rhoL_p\leq\rho_h$ if $A$ is diagonalizable.
\end{itemize}
}
\end{prop}

\smallskip

\begin{proof}
Statement (a) follows by noting that $\rhoU$ is the robust instability
radius with respect to homogeneous perturbation $\Delta=\delta I$.
The former part of statement (b) is from Inoue {\em et al.} \cite{Inoue:ECC2013}.
\end{proof}
\section{Cyclic Network with Multiplicative Uncertainties}
\label{sec:GeneralCycNetwork}

This section provides some lower bounds on the RIRs  for the case of  
cyclic network structure with multiplicative uncertainties. 

\subsection{Equivalence of homogeneous and heterogeneous perturbations}
\label{subsec:EquiHeteroHomo}

Consider a set of uncertain $n$ agents $(1+\delta_i(s))h(s)$, $i=1,\ldots,n$, with cyclic connections. 
This is a special case of (\ref{sys}) where $G(s)$ belongs to %
\begin{eqnarray} \label{Gc} 
&& \calG_c := 
\{  G(s) = {\LFT(H_o(s)\otimes I_n,A) } \; | \;\;  \nonumber \\
&&
\hspace*{-12mm}
H_o(s) = \begin{mat}{cc} 0 & h(s) \\ 1 & h(s) \end{mat}, \;
A=\begin{mat}{cc} o\t & -\mu \\ -\mu I_{n-1} & o \end{mat} \} , 
\end{eqnarray}
where $\mu\in\IR$ is a positive scalar representing the strength of the interaction 
and $o\in\IR^{n-1}$ is a zero vector.
Note that the characteristic equation for the feedback system of  
$G(s) \in \calG_c$ and $\Delta\in\bfDel$ is given by
\begin{equation} \label{checyc}
1+\mu^nh(s)^n\prod_{i=1}^n(1+\delta_i(s))=0.
\end{equation}
The following result is instrumental for reducing the instability analysis 
of the interconnected system into that of individual subsystems. 

\begin{lemma} \label{lem:convex}
{\em 
Let an integer $n>0$ and a real number $r>0$ be given and consider 
the following sets of complex numbers:
\[
\IW:=\{\,(1+\delta)^n:\,\delta\in\bfdel_c\,\}, 
\]
 \[
\IV:=\{\,\prod_{i=1}^n(1+\delta_i):\,\delta_i\in\bfdel_c\,\}, 
\]
\[
\bfdel_c:=\{\delta\in\IC:\,|\delta|\leq r\,\}.
\]
Then $\IW=\IV$.
}
\end{lemma}
\begin{proof}
The inclusion relationship $\IW\subseteq\IV$ is obvious. We prove
$\IV\subseteq\IW$ by showing that, for arbitrarily chosen
$\delta_i\in\bfdel_c$, $i=1,2,\ldots,n$, there exist $\delta\in\bfdel_c$ 
such that
\begin{equation} \label{replace}
(1+\delta)^n=\prod_{i=1}^n(1+\delta_i),
\end{equation}
or equivalently,
\[
\log(1+\delta)=\frac{1}{n}\sum_{i=1}^n\log(1+\delta_i).
\]
This is the case because the average of arbitrary $n$ elements of the set
$\bfThe:=\log(1+\bfdel_c)$ belongs to the set $\bfThe$ due to its convexity
from Lemma \ref{lem:LogConvex} shown below without proof.
\end{proof} 

\medskip

\begin{lemma} \label{lem:LogConvex}
{\em 
Let $r\in\IR$ be given such that $0<r<1$, and consider\footnote{
The function $\log(z)$ for $z\in\IC$, not negative real, is the principal 
logarithm defined as
\[
\log(z):=\log(a)+j\theta
\]
where $a,\theta\in\IR$ are uniquely determined from 
\[
z=ae^{j\theta}, \hs a>0, \hs \theta\in(-\pi,\pi).
\]
}
\[
\bfThe:=\log(1+\bfdel), \hs
\bfdel:=\{\delta\in\IC:\,|\delta|\leq r\,\}.
\]
The set $\bfThe$ is convex.
}
\end{lemma}

\medskip

Lemma~\ref{lem:convex} essentially establishes the equivalence between the
heterogeneous and homogeneous perturbations. Replacing $\IV$ by $\IW$, 
the characteristic equation in (\ref{checyc}) becomes
\[
(-\mu)^nh(s)^n(1+\delta(s))^n=1.
\]
Recalling that $n$ is odd and taking the $(1/n)^{\rm th}$ power, 
\[
\hspace*{-2mm}
\lambda_ih(s)(1+\delta(s))=1, \; 
\lambda_i:=\mu e^{j(2i-1)\pi/n}, \; i\in\II_n.
\]
Rearranging,
\[
\hspace*{-2mm}
1-\delta(s)g_i(s)=0, \hs g_i(s):=\frac{\lambda_ih(s)}{1-\lambda_ih(s)},
\; i\in\II_n,
\]
which is the characteristic equation in (\ref{rhou}) obtained by assuming
that the perturbation is homogeneous, i.e., $\Delta=\delta I$ and noting
that $\lambda_i$ are the eigenvalues of $A$. A subtle issue is that 
$\delta(s)$ satisfying (\ref{replace}) is not rational in general even if
$\delta_i\in\RHinf$. Thus, Lemma~\ref{lem:convex} establishes the equivalence
of homogeneous and heterogeneous perturbations in the context of robust
instability analysis when the class of perturbations is extended to include
irrational transfer functions. For the rational case of our interest, 
Lemma~\ref{lem:convex} is useful for characterizing a lower bound on the
robust instability radius.

\subsection{Lower bound on robust instability radius}
\label{subsec:LoweBound}

We here present some lower bounds on the RIRs based on the convexity property 
shown in the previous subsection for multi-agent dynamical systems for $G(s) \in \calG_c$.   

\medskip

\begin{theorem} \label{thm:cyclb}
{\em 
Consider the cyclic multi-agent dynamical systems  $G(s) \in \calG_c$ with $\Delta\in\bfDel$. 
We assume the system is nominally strictly unstable\footnote{
A system is called strictly unstable if it is unstable with at least one characteristic root 
in the open right half plane.
}, 
and the number of subsystems $n$ is odd. 
We consider the following three classes of perturbations: 
(i) dynamic heterogeneous uncertainty $\Delta\in\bfDel$ with (\ref{bfDel}), 
(ii) dynamic homogeneous uncertainty $\Delta=\delta I$ with $\delta\in\RHinf$, and 
(iii) parametric complex uncertainty $\Delta\in\bfDel_c$.  
Define the corresponding RIRs $\rho_*$, $\rho_h$, and $\rho_c$, by (\ref{rir}), (\ref{hrir}), 
and (\ref{crir}), respectively.  
Then $\rhoL_+$ is a lower bound on $\rho_*$, 
$\rho_h$, and $\rho_c$, {\it i.e.},
\begin{equation} \label{LowerBounds}
 \rhoL_+ \leq \rho_*, \hs \rhoL_+ \leq \rho_h, \hs \rhoL_+ \leq \rho_c ,
\end{equation}
where
\begin{equation} \label{varrho}
\rhoL_+:=\max_{k\in\IU}\frac{1}{\|g_k\|_{L_\infty}}, \hs
\end{equation}
\[
\hspace*{-4mm}
g_k(s):=\frac{\lambda_kh(s)}{1-\lambda_kh(s)}, \;
\lambda_k:=\mu e^{j(2k-1)\pi/n}, \hs k\in\II_n.
\]
and $\IU\subseteq\II_n$ is the set of indices $k$ such that $g_k(s)$ is unstable. 
}
\end{theorem} 

\smallskip

Although the detailed proof is omitted due to the page limitation,  
the proof is done by contradiction using Lemma \ref{lem:convex}.

\section{Exact RIR Analysis}
\label{sec:ExactRIR}

This section will show that we can get the RIR for first order time-lag systems by a constructive way 
proposed in \cite{HIH:LCSS2020, HIH:Automatica2020}, where the first step is to find a marginally stabilizing 
all-pass function and the second step is a technique of $\epsilon$ perturbation for the complete stabilization. 

\subsection{First order time-lag systems}
\label{subsec:1stOrderExact}

Consider the first order time-lag system represented by  
\begin{equation} \label{1stLag}
h(s):=\frac{K}{\tau s+1}, \hs K>0, \; \tau>0 . 
\end{equation}
Fortunately, we can get the exact RIR for the system as follows. 

\medskip

\begin{prop} \label{prop:1st}
{\em 
Consider the cyclic multi-agent dynamical systems  $G(s) \in \calG_c$ 
with the first order $h(s)$ in (\ref{1stLag}) and $\Delta\in\bfDel$. 
We assume the system is nominally strictly unstable, i.e.,  
$K < \mu \cos{\theta_n}; \;  \theta_n  := \pi/n$  and the number of subsystems $n$ is odd. 
Then, the robust instability radius $\rho_*$  is given by  
\begin{equation} \label{ExactRIR:1stLag}
  \rho_* = \rhoL_+ = 1/\|g_1\|_{L_\infty} = 1 - K /(\mu \cos{\theta_n} ). 
\end{equation}
}
\end{prop}

\smallskip

The proof is omitted due to the page limitation. 

\smallskip

We here illustrate the situation by using the 
inverse Nyquist plot of $h(s)$, i.e., $\phi(j\omega) = 1/h(j\omega)$, and its perturbed version.  

Consider an example with $n=9$ , $\mu=3$, $K=1$, and $\tau=1$. 
The Nyquist plot of $\tilde\phi_o(s)$ (see the appendix for its definition) 
is plotted as the red curve in Fig.~\ref{fig:horaku1}. 
The stability region for $\tilde\phi_o(s)$ is given by the region to the left of the Nyquist plot 
$\tilde\phi_o(j\omega)$ indicated by the red curve.
The yellow region indicates the value set 
\[
\{\phi(j\omega)/(1+\delta):\, \omega\in\IR, ~ \delta\in\IC, ~ 
|\delta|\leq \rhoL_+~\}.
\]
The point $\lambda_1=3e^{j\pi/9}$ lies
on the right boundary of the yellow region (blue star). 
The gain $|\tilde\phi_o(j\omega)|$ monotonically increases
with $\omega>0$. This property implies that if 
$\mu e^{j\pi/n}$ is on the Nyquist plot $\tilde\phi_o(j\omega)$, then all the
points $\mu e^{j(k/n)\pi}$ for odd $1<k\leq n$ lie to the left of the Nyquist
plot (i.e. the stability region). To see this, note that the portion of the
Nyquist plot for $\omega>\omega_p$ is outside of the circle $\mu e^{j\theta}$,
where the lower bound frequency is defined by 
$\mu e^{j\pi/n}=\tilde\phi_o(j\omega_p)$.
Likewise, the low frequency portion $\omega<\omega_p$ is inside the circle.
This implies that the inverse Nyquist plot $\tilde\phi_\varepsilon(s)$  satisfies the stability 
condition meaning that the perturbed network with 
$\tilde h_\varepsilon(s) := 1/\tilde\phi_\varepsilon(s)$ is stable.

\begin{figure}[h]
\centerline{\epsfig{file=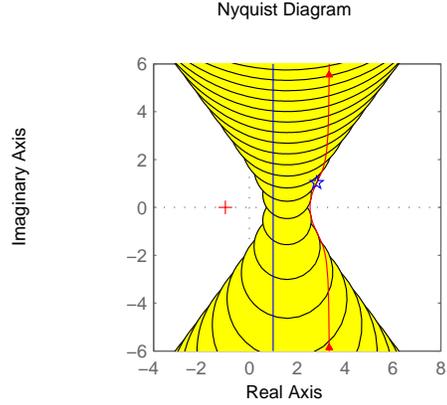,width=60mm}}
\caption{Nyquist plots: $\phi(j\omega)$ is blue, 
$\tilde\phi_o(j\omega)$ is red. }
\label{fig:horaku1}
\end{figure}
\subsection{Possible Extension}
\label{subsec:Extension}

The idea for the derivation of the exact RIR for first order time-lag systems has a potential to 
be applicable to some classes of higher order systems including second order time-lag systems.  
To show this, we give a numerical example with $h(s) = 3/(s+1)(s+3)$  and  $\mu = 5$. 

We only consider the case of odd $n$, where the eigenvalues of $A$ are at 
$\lambda_k=5e^{j(k/n)\pi}$ for $k=1,3,\ldots,2n-1$.
From Proposition~\ref{prop:ul}, the upper and lower bounds on the
robust instability radius are given by $\rhoU$ and $\rhoL_p$ with
\[
g_i(s):=\frac{\lambda_i h(s)}{1-\lambda_i h(s)} 
= \frac{3 \lambda_i}{(s+1)(s+3) - 3 \lambda_i} . 
\]

Numerical results of robust instability radius for odd $n$ are plotted in Fig.~\ref{fig:rir}.  
The lower bound $\rhoL_p$ defined by (\ref{rhol}) is shown by the blue curve.  
The exact values for the dynamic RIR $\rho_*$ as well as for the complex parametric RIR $\rho_c$ 
are computable for this example based on the results presented later, and are indicated 
by the red and green curves. In fact, $\rho_*=\rho_h=\rhoU$ for this example. 
The exact values are based on the eigenvalue $\lambda_1$ that is the closest 
to the real axis. The lower bound $\rhoL_p$ is based on the
eigenvalue $\lambda_k$ closest to the Nyquist plot $\phi(j\omega)$.
Although complex parametric uncertainty and real rational dynamic uncertainty are equivalent 
for robust stability, they turned out not to be equivalent for robust instability as shown by the gap 
in the red and green curves. 
The property that $\rho_c$ is an upper bound on $\rho_*$ is not a coincidence for this example.

\begin{figure}[h]
\centerline{\epsfig{file=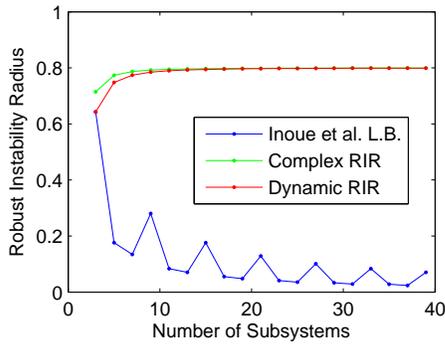,width=60mm}}
\caption{Robust instability radius. }
\label{fig:rir}
\end{figure}
\section{Conclusion}
\label{sec:Concl}

This paper proposed a robust instability problem for multi-agent dynamical systems with cyclic structure. 
It has been shown that we can get the exact robust instability radius (RIR) for first order time-lag systems. 
The main contribution is that we can get the exact value of maximum allowable dynamic perturbation 
for the persistence of instability for multi-agent dynamical systems with any number of agents.  
This is a clear contrast with the previous result in \cite{HIH:Automatica2020}, which is difficult to apply 
to large-scale network  systems.  
Although this paper only provided  the first step results, the approach has a potential to apply to the wider 
class of subsystems $h(s)$ as briefly explained in Subsection \ref{subsec:Extension}.   

The future work along this research direction includes 
(i) a characterization of class of systems for which the RIR can be analyzed exactly, 
(ii) an extension to the robust instability margin analysis for a class of nonlinear systems by taking the change of the 
equilibrium point into account, and  
(iii) applications to a more general type of biomolecular systems than Repressilator investigated in \cite{HIH:Automatica2020}. 

\smallskip
 
\noindent
{\bf Acknowledgments:}  $\;$ 
This work was supported in part by the Ministry of Education, Culture, Sports, Science and Technology in Japan 
through Grant-in-Aid for Scientific Research (A) 21246067 and (B) 18H01464.


\begin{thebibliography}{1}
%

\bibitem{Alon2006}
U.~Alon,
 \newblock {\em An Introduction to Systems Biology: Design Principles of
	 Biological Circuits},
\newblock Chapman and Hall/CRC, 2006.

\bibitem{Elowitz2000}
M.~B. Elowitz and S. Leibler, ``A synthetic oscillatory network of
	 transcriptional regulators,''
	 \emph{Nature}, vol. 403, no. 6767, pp. 335--338, 2000.

\bibitem{FHNmodel}
R.~FitzHugh,
``Impulses and physiological states in theoretical models of nerve membrane,''
{\em Biophys. J.}, Vol.1, pp.445-466, 1961.

\bibitem{HIH:Automatica2019}
S. Hara, T. Iwasaki, and Y. Hori,  
``Robust stability analysis for LTI systems with generalized frequency variables and its application to gene regulatory networks,''
{\em Automatica}, vol. 105, no. 9, pp. 96-106, 2019. 

\bibitem{HIH:LCSS2020}
S. Hara, T. Iwasaki, and Y. Hori,  
``Robust instability analysis with neuronal dynamics,''
To appear in {\em IEEE Conf. Dec. Contr.}, 2020 
(arXiv 2003.01868).

\bibitem{HIH:Automatica2020}
S. Hara, T. Iwasaki, and Y. Hori,  
``Robust Instability Margin Analysis with Application to Biological Oscillations,''
submitted to {\em Automatica}, 2020 
(arXiv 2008.10983).

\bibitem{HTI:TAC2014}
S. Hara, H. Tanaka, and T. Iwasaki, 
``Stability analysis of systems with generalized frequency variables,''
{\em IEEE Trans. on Automatic Control}, Vol. 59, No. 2, pp. 313-326, 2014.

\bibitem{YMKH:MBMC2015}
Y.~Hori, H.~Miyazako, S.~Kumagai,S.~Hara, 
\newblock ``Coordinated spatial pattern formation in biomolecular communication networks,''
\newblock {\em IEEE Trans. on Molecular, Biological and Multi-Scale Communications}, 
vol.~1 no. 2, pp. 111--121, 2015.

\bibitem{Inoue:ECC2013}
M. Inoue, et al., 
``An instability condition for uncertain systems toward robust bifurcation analysis,''
{\em Proc. of the European Control Conf. 2013}, pp. 3264-3269, 2013.
     
\bibitem{Tsypkin:Automatica1994}
Y.~Z.~Tsypkin, D.~ J.~ Hill, and A.~J.~ Isaksson, 
"A frequency-domain robust instability criterion for time-varying and non-linear systems," 
{\em Automatica}, vol. 30, no. 11, pp. 1779-1783, 1994. 

\bibitem{Youla:Automatica1974}
D.C. Youla, J.J. Bongiorno, Jr., and C.N. Lu, 
``Single-loop feedback-stabilization of linear multivariable dynamical plants,''
{\em Automatica}, vol.10, pp.159-173, 1974.

%
\end{thebibliography}
\end{document}